\documentclass[preprint,12pt]{elsarticle}
\usepackage{amsfonts}
\usepackage{graphicx}
\usepackage{subfigure}
\usepackage{}
\usepackage[colorlinks,linkcolor=red,anchorcolor=blue,citecolor=green]{hyperref}
\usepackage{color}
\usepackage{mathrsfs}
\usepackage{amssymb}
\usepackage{mathrsfs,amsthm,amsmath,amssymb}
\usepackage{enumerate}
\usepackage{array}
\usepackage{appendix}
\usepackage{accents}
\numberwithin{equation}{section}
\allowdisplaybreaks[4]
\allowdisplaybreaks[4]
\renewenvironment{proof}{\noindent{\bf Proof. }}{\hfill $\Box$}
\usepackage{CJK}
\usepackage{indentfirst}
\usepackage{amsmath}
\newtheorem{theorem}{Theorem}[section]
\newtheorem{lemma}{Lemma}[section]

\newtheorem{theo}{Theorem}[section]
\newtheorem{pr}{Proposition}[section]
\newtheorem{lem}{Lemma}[section]
\newtheorem{co}{Corollary}[section]
\newtheorem{re}{Remark}[section]

\newcommand{\R}{\mathbb{R}}

\newcommand{\Pro}{\mathbb{P}}
\newcommand{\Exp}{\mathbb{E}}

\newcommand{\ol}{\overline}
\newcommand{\td}{\tilde}
\newcommand{\df}{\,\mathrm{d}}
\newcommand{\bt}{\begin{theo}}
	\newcommand{\et}{\end{theo}}
\newcommand{\bp}{\begin{pr}}
	\newcommand{\ep}{\end{pr}}
\newcommand{\bl}{\begin{lem}}
	\newcommand{\el}{\end{lem}}
\newcommand{\bc}{\begin{co}}
	\newcommand{\ec}{\end{co}}
\newcommand{\br}{\begin{re}}
	\newcommand{\er}{\end{re}}
\newcommand{\be}{\begin{eqnarray}}
\newcommand{\ee}{\end{eqnarray}}
\newcommand{\by}{\begin{eqnarray*}}
	\newcommand{\ey}{\end{eqnarray*}}
\newcommand{\bn}{\begin{enumerate}}
	\newcommand{\en}{\end{enumerate}}

\begin{document}
	\begin{frontmatter}
		
		
		
		\title{Geometric Brownian motion with affine drift and its time-integral
			\tnoteref{label4}}
		\tnotetext[label4]{R. Feng is supported by an endowment from the State Farm Companies Foundation. P. Jiang is supported by the  Postdoctoral Science Foundation of China (No.2020M671853) and the National Natural Science Foundation of China (No. 11631004, 71532001).}
		
		\author[label1]{Runhuan Feng}
		\address[label1]{Department of Mathematics, University of Illinois at Urbana-Champaign, Illinois, USA}
		\author[label2]{Pingping Jiang\corref{cor1}}
		\address[label2]{School of Management and Economics,The Chinese University of Hong Kong, Shenzhen, Shenzhen, Guangdong, 518172, P.R. China\\ 
			School of Management, University of Science and Technology of China, Hefei, Anhui, 230026, P.R.China}
		
		\author[label3]{Hans Volkmer}
		\address[label3]{Department of Mathematical Sciences, University of Wisconsin--Milwaukee, Wisconsin, USA}
		\cortext[cor1]{Corresponding author\\
			\indent
			Email address:
			rfeng@illinois.edu (R. Feng);
			jiangpingping@cuhk.edu.cn (P. Jiang);\\
			volkmer@uwm.edu (H. Volkmer).}

		\begin{abstract}
			The joint distribution of a geometric Brownian motion and its time-integral was derived in a seminal paper by Yor (1992) using Lamperti's transformation, leading to explicit solutions in terms of modified Bessel functions. 
			In this paper, we revisit this classic result using the simple Laplace transform approach in connection to the Heun differential equation.  
		We extend the methodology to the geometric Brownian motion with affine drift and show that the joint distribution of this process and its time-integral can be determined by a doubly-confluent Heun equation. Furthermore, the joint Laplace transform of the process and its time-integral is derived from the asymptotics of the solutions.  
		{ 
		In addition, we provide an application by using the results for the asymptotics of the 
		double-confluent Heun equation in pricing Asian options. 
		Numerical results show the accuracy and efficiency of this new method.}
		\end{abstract}
		
		\begin{keyword}
			Doubly-confluent Heun equation, geometric Brownian motion with affine drift, Lamperti's transformation, asymptotics,  boundary value problem.
		\end{keyword}
		
	\end{frontmatter}
	
	\section{Introduction}
	Heun's differential equation (c.f. Ronveaux \cite{Ron}) is
	\begin{equation}\label{I:Heun}
	\frac{d^2y}{dz^2}+\left(\frac{\gamma}{z}+\frac{\delta}{z-1}+\frac{\epsilon}{z-a}\right)\frac{dy}{dz}+\frac{\eta z-q}{z(z-1)(z-a)} y=0,
	\end{equation}
	where $\gamma,\delta,\epsilon, \eta, q$ and $a\ne0,1$ are complex parameters.
	There are four regular singularities at $z=0,1,a,\infty$.
	Heun's equation has four confluent forms, the confluent, doubly-confluent, biconfluent and triconfluent Heun equations.
	In recent years the Heun equation and its confluent forms have found many applications in natural sciences. 
In the theory of black holes, the perturbation equations of massless fields for the Kerr-de Sitter geometry can be written in the form of separable equations. The equations have five definite singularities so that the analysis has been expected to be difficult. Suzuki et al. \cite{STU} showed that these equations can be transformed to Heun's equations thus the known technique could be used for the analysis of the solutions. 
	As Schr\"odinger equation for harmonium and related models may be transformed to the biconfluent Heun equation, Karwowski and Witek \cite{KW} discussed the solubility of  this equation and its applications in quantum chemistry.
Chugunova and Volkmer \cite{CH} also investigated that the set of eigenvalues of a non-self-adjoint differential operator arising in fluid dynamics were related to the eigenvalues of Heun's differential equation.

	In this paper we present applications of the doubly-confluent Heun equation in stochastic analysis.
	The doubly confluent Heun equation can be obtained from \eqref{I:Heun} as follows. First, by substitution $x=h z$, we move the singularity $z=1$ to $x=h$.
	The equation becomes
	\begin{equation}\label{I:Heun2}
	x(x-h)(x-A)\frac{d^2y}{dx^2}+\left(c x(x-A)+d(x-A)-\nu A x(x-h)\right)\frac{dy}{dx}+A(\tau x+ b)y=0,
	\end{equation}
	where
	\[ A=ah,\quad c=\gamma+\delta,\quad d=-\gamma h,\quad \epsilon=-\nu A,\quad \eta=\tau A,\quad hq=-bA .\]
	Dividing \eqref{I:Heun2} by $-A$ and letting $h\to0$, $A\to\infty$ (double confluence), we obtain
	\begin{equation}\label{I:Heun3}
	x^2\frac{d^2 y}{dx^2}+\left(\nu x^2+c x+d\right)\frac{dy}{dx}-(\tau x+b)y =0.
	\end{equation}
A study of equation \eqref{I:Heun3} can be found in Ronveaux~\cite[Part C]{Ron}.
In this paper, we consider this equation for the special case where $\nu=0$,
	\begin{equation}\label{I:Heun4}
	x^2\frac{d^2 y}{dx^2}+\left(c x+d\right)\frac{dy}{dx}-(\tau x+b)y =0.
	\end{equation}
It should be pointed out that, while this equation appeared in the literature as a special case,  its solutions were not studied in previous literature.

In this paper, we explore the close relationship between the doubly-confluent Heun equation \eqref{I:Heun4} and the geometric Brownian  motion with affine drift.  We mainly focus on this process whose infinitesimal generator is given by
\begin{align*}
(\mathcal{G}f)(x)=2x^2f''(x)+[2(\nu+1)x+1]f'(x),
\end{align*}
where $\nu \in \mathbb{R}$.
This process has many natural applications in finance and insurance. For example, it can be used to model  the price dynamics of dividend-paying stocks with constant dividend rate, see Lewis~\cite{Lew} for more details. 
While it is not an absolute representation of reality, the model can be used to approximate the dynamics of actual stocks with dividends paid at discrete points, as an improvement to the classical Black-Scholes-Merton model.	

In the classical Black-Scholes model for option pricing, stock prices are modeled by geometric Brownian motions. Among the best studied exotic options is the Asian option, whose payoff depends on the average of stock prices. In a continuous-time model, the average price is modeled by the integral of geometric Brownian motion over the life of the option divided by the length of that period. Motivated by such a pricing problem, Yor \cite{Yor92} developed a series of papers on exponential functionals of Brownian motion $B_t$. At the center of focus for financial application is the joint distribution of $(\exp\{B^{(\nu)}_t\}, A^{(\nu)}_t)$ where
	\begin{equation}\label{J:BA}
	B^{(\nu)}_t:=\nu t +B_t,\qquad A^{(\nu)}_t:=\int^t_0 \exp\{2B^{(\nu)}_s\} \df s.
	\end{equation}
	The derivation of the joint distribution was based on a combination of time change, change of measures and their connections to Bessel processes. However, in more sophisticated diffusion models, such as the geometric Brownian  motion with affine drift,  a generalization of such techniques seems to be less fruitful.  Hence, in this paper, we use the connection between the joint distribution of time-homogeneous diffusion and its time-integral and differential equations to develop computational methods. Linetsky~\cite{Lin04b} also used the connection between the time integral of stock prices and the geometric Brownian motion with affine drift for pricing Asian options. The process is also used in the risk management of an insurer's net liability in variable annuity guaranteed benefits in Feng and Volkmer~\cite{FV0, FV, FV2} and Feng and Jing~\cite{FJ}.
	
	The rest of the paper is organized as follows. In Section 2, we introduce particular solutions to the Heun equation \eqref{I:Heun4} that are required for applications to stochastic analysis for various applications in finance.
While the Heun equation is already known in analysis, the particular solutions presented here are not known previously in the literature.  The paper also proposes a method to compute these solutions to high accuracy.
	It is shown in Section 3 that the joint distribution of a solution $X_t$ to a stochastic differential equation and its integral $Y_t=\int_0^t X_s\,ds$
	is determined by appropriate solutions to a linear differential equation of the second order.
	In Section 4, the special case of $X_t$ being a geometric Brownian motion with affine drift is discussed and its connection to the Heun equation \eqref{I:Heun4} is revealed. Results from Sections 2 are used to develop a numerical method to compute the joint distribution of $X_t$ and $Y_t$. Section 5 is dedicated to a new stochastic process $Z_t$ resulting from the Lamperti transformation $X_t=Z_{Y_t}$. The distribution of $Z_t$ can again be computed by solving equation \eqref{I:Heun4}. 
	An explicit expression for the joint Laplace transform of the geometric Brownian motion with affine drift and its time integral is provided in Section 6, which relies on asymptotics of a solution to a boundary value problem involving an inhomogeneous equation corresponding to \eqref{I:Heun4}.
	{ By using the results, we provide an application in pricing Asian options on divided-paying stocks. 
	Numerical results show the accuracy and efficiency of this new method.}
	
	\section{A doubly confluent Heun equation}
	Let us consider differential equation \eqref{I:Heun4} (with $\tau$ replaced by $a$)
	\begin{equation}\label{DCHE:ode}
	x^2 y''(x)+(cx+d)y'(x)-(ax+b)y(x) =0, \quad 0<x<\infty,
	\end{equation}
	containing four real parameters $a,b,c,d$.
	The equation has two singularities: $x=0$ and $x=\infty$.
	If $d\ne 0$ then $x=0$ is an irregular singularity. If $d=0$ then $x=0$ is a regular singularity, and the general solution of \eqref{DCHE:ode} is
	\begin{equation}\label{DCHE:d0}
	y(x)=x^{\frac12(1-c)}\left\{ C_1I_\lambda(2\sqrt{ax})+C_2 K_\lambda(2\sqrt{ax})\right\} ,
	\end{equation}
	where $I_\lambda, K_\lambda$ denote modified Bessel functions, and
	\be \lambda=\sqrt{(c-1)^2+4b} . \label{DCHE:lamb}\ee
	If $a\ne 0$ then $x=\infty$ is an irregular singularity. If $a=0$
	then $x=\infty$ is a regular singularity, and the general solution of \eqref{DCHE:ode} is
	\begin{equation}\label{DCHE:a0}
	y(x)=x^{-\mu} \left\{C_1 M\left(\mu,1+\lambda,\tfrac{d}{x}\right)+
	C_2 U\left(\mu,1+\lambda,\tfrac{d}{x}\right)\right\},
	\end{equation}
	where
	\[ \mu=\tfrac12(c-1+\lambda), \]
	and $M,U$ denote Kummer functions (c.f. Olver et al.~\cite[Chapter 13]{NIST}).
	
	We now turn to the general case $d\ne 0$ and $a\ne 0$.
	We will apply methods of Olver~\cite[Chapter 7]{Olv} to introduce solutions of \eqref{DCHE:ode}. We are interested in finding a recessive solution $y_1(x)$ at $x=0$,  and
	a recessive solution $y_2(x)$ at $x=\infty$, that is, solutions with the properties that $y_1(x)/y(x)\to 0$ as $x\to 0^+$ and $y_2(x)/y(x)\to 0$ as $x\to\infty$ for every solution $y$ which is linearly independent of $y_1, y_2$, respectively.
	
	\bt \label{DCHE:y2}
	Suppose $a>0$. The differential equation \eqref{DCHE:ode} has a unique solution $y_2$ on $(0,\infty)$ such that, as $x\rightarrow \infty$,
	\be y_2(x)\sim x^{\frac14-\frac{c}2} \exp(-2\sqrt{ax}) \sum^\infty_{k=0} A_k\, x^{-\frac{k}2},\label{DCHE:y2asy}\ee
	where $A_0=1$ and $A_k$'s are determined by the following recursion for $k=1,2, \cdots$.
	\be A_k=\frac{4b-(k+\frac12-c)(k-\frac32+c)}{4\sqrt{a}k}A_{k-1}+ \frac{d}{k}A_{k-2},\qquad \label{DCHE:rec}\ee
	with the understanding that $A_{-1}=0$.
	The asymptotic formula \eqref{DCHE:y2asy} may be differentiated term-by-term.
	If $b\ge 0$ then the solution $y_2$ is positive and decreasing.
	\et
	\begin{proof}
		Chapter 7 of Olver~\cite{Olv} considers the differential equation
		\[w''(z)+f(z)w'(z)+g(z)w(z)=0,\]
		for which the singularity is located at infinity and
		\[f(z)=\sum^\infty_{k=0} \frac{f_k}{z^k},\qquad g(z)=\sum^\infty_{k=0} \frac{g_k}{z^k}.\]
		The first goal is to construct a formal solution of the form
		\begin{equation}\label{DCHE:formal}
		w(z)=e^{\lambda z}z^\mu \sum_{k=0}^\infty A_k z^{-k} .
		\end{equation}
		In the case of \eqref{DCHE:ode}, we have
		\[f_1=c, f_2=d,g_1=-a, g_2=-b,\] and all other $f_i$'s and $g_i$'s are zero. This falls into the category where $f^2_0=4g_0$ and Fabry's transformation should be used. Let $z=x^{1/2}$ and $y(x)=w(z)$. Then \eqref{DCHE:ode} becomes
		\begin{equation}\label{DCHE:Heun2}
		w''(z)+\left(\frac{2c-1}{z}+\frac{2d}{z^3}\right)w'(z)-4\left(a+\frac{b}{z^2}\right)w(z) =0 .
		\end{equation}
		We have
		\[ f_1=2c-1,\, f_3=2d,\, g_0=-4a,\, g_2=-4b ,\]
		and all other $f_i$'s and $g_i$'s are zero.
		The roots of $\lambda^2+f_0\lambda+g_0=0$ are $\lambda=\pm 2 a^{1/2}$. Since we assumed that $a>0$, these solutions are distinct.
		We are interested in a recessive solution of \eqref{DCHE:Heun2} at $x=\infty$, so we pick $\lambda=-2a^{1/2}$. We compute $\mu=\frac12-c$ from Olver~\cite[page 230, (1.10)]{Olv}.
		Therefore, we have a formal solution of \eqref{DCHE:Heun2} of the form \eqref{DCHE:formal}
		which translates to the right-hand side of \eqref{DCHE:y2asy}.
		The coefficients $A_k$ are computed from  $A_0=1$ and recursively from Olver~\cite[(1.11)]{Olv}, which can be written as \eqref{DCHE:rec}.
		By Olver~\cite[Theorem 2.1, page 232]{Olv}, the ODE \eqref{DCHE:ode} has a solution $y_2(x)$ with the asymptotic behavior \eqref{DCHE:y2asy}. Since it is the recessive solution, $y_2(x)$ is uniquely determined by \eqref{DCHE:y2asy}.
		
		Actually, \eqref{DCHE:y2asy} holds in an open sector of the complex plane containing the positive real axis. This justifies term-by-term differentiation in \eqref{DCHE:y2asy}.
		
		We now prove that the solution $y_2$ is positive and decreasing if $b\ge 0$.  For every $x_1>0$ there is $x_2>x_1$ such that $y_2(x_2)>0$ and $y_2'(x_2)<0$.
		If $y_2$ is not decreasing on $(0,x_2)$ there would be a point $x_3<x_2$ such that $y'(x_3)=0$ and $y_2$ is decreasing on $(x_3,x_2)$.
		This is impossible because \eqref{DCHE:ode} under the assumption $a>0,b\ge 0$ implies $y''(x_3)>0$. As $x_1$ is arbitrary, $y_2(x)$ is positive and decreasing for all $x\in (0,\infty)$.
	\end{proof}
	
	\br
	When $d=0$ the recursive formula \eqref{DCHE:rec} can be simplified to
	\[A_k=\frac1{2\sqrt{a}} \frac{a_k(\lambda)}{a_{k-1}(\lambda)}A_{k-1},\qquad k=1,2,\cdots,\] where $\lambda$ is given by \eqref{DCHE:lamb} and
	\[a_k(\lambda)=\frac{(4\lambda^2-1^2)(4\lambda^2-3^2)\cdots(4\lambda^2-(2k-1)^2)}{k! 8^k}.\]Hence, we get
	\[A_k=\frac1{(2\sqrt{a})^k} a_k(\lambda).\] Substituting the expression for $A_k$ in \eqref{DCHE:y2asy} yields, as $x\rightarrow \infty$,
	\be y_2(x)\sim x^{1/4-c/2} \exp (-2\sqrt{ax})\sum^\infty_{k=1} a_k(\lambda) (2\sqrt{ax})^{-k}.\label{DCHE:d0asym}\ee
	Obviously, apart from a constant multiple, $y_2(x)$ is equal to
	$x^{\frac12(1-c)}K_\lambda (2\sqrt{ax})$.
	and \eqref{DCHE:d0asym} agrees with Hankel's expansion (Olver et al.~\cite[(10.40.2)]{NIST}),
	\[K_\lambda(z)\sim \left( \frac{\pi}{2z}\right)^{-\frac12} e^{-z} \sum^\infty_{k=1} \frac{a_k(\lambda)}{z^k}.\]
	\er
	
	\bt \label{DCHE:y1}
	If $d>0$, the differential equation \eqref{DCHE:ode} has a unique solution $\hat{y}_1$ on $(0,\infty)$ such that, as $x\to0^+$,
	\be \hat{y}_1(x)\sim  \sum^\infty_{k=0} A_k\, x^{k},\label{DCHE:y1asy}\ee
	where $A_0=1$ and $A_k$'s are determined by the following recursion for $k=1,2, \cdots$.
	\be  A_k=\frac{b+(k+c-2)(1-k)}{k d}A_{k-1} +\frac{a}{kd} A_{k-2},\qquad \label{DCHE:rec2}\ee
	with the understanding that $A_{-1}=0$.  If $d<0$, the equation \eqref{DCHE:ode} has a unique solution $y_1$ on $(0,\infty)$ such that, as $x\to0^+$,
	\be y_1(x)\sim  \exp\left( \frac{d}{x} \right)x^{2-c}\sum^\infty_{k=0} A_k\, x^{k},\label{DCHE:y1asy2}\ee
	where $A_0=1$ and $A_s$ are determined by the following recursion for $k=1,2, \cdots$.
	\be  A_k=\frac{(k-c+1)k-b}{kd}A_{k-1} -\frac{a}{kd} A_{k-2},\qquad \label{DCHE:rec3}\ee
	with the understanding that $A_{-1}=0$.
	The asymptotic formulas \eqref{DCHE:y1asy}, \eqref{DCHE:y1asy2} may be differentiated term-by-term.
	If $a>0, b\ge 0$ then $y_1$ is positive and increasing.
	\et
	\begin{proof}
		We substitute $z=1/x$, $y(x)=w(z)$ in \eqref{DCHE:ode} and obtain
		\[w''(z)+\left( \frac{2-c}{z}-d\right) w'(z)-\left( \frac{a}{z^3}+\frac{b}{z^2}\right) w(z)=0,\qquad 0<z<\infty,\]
		for which the singularity at $z=\infty$ corresponds to the singularity at $x=0$ in \eqref{DCHE:ode}. We now argue as in the proof of Theorem \ref{DCHE:y2}.
		Now we have
		\[f_1=2-c, f_0=-d, g_2=-b, g_3=-a,\]
		and all other $f_i$'s and $g_i$'s are zero.
		
		If $d>0$, we consider a formal solution of the form \eqref{DCHE:formal} where $\lambda=0$, $\mu=0$ and $A_k$'s satisfy the recursion \eqref{DCHE:rec2}.
		We apply Olver~\cite[Theorem 2.1, page 232]{Olv} and obtain a (recessive) solution $\hat{y}_1$ satisfying \eqref{DCHE:y1asy}.
		It follows that  if $b>0$, as $x\rightarrow 0^+$,
		\[\hat{y}_1(x)\sim 1+\frac{b}{d} x.\] Or if $b=0$, we have as $x\rightarrow 0^+$,
		\[\hat{y}_1(x)\sim 1+\frac{a}{2d} x^2. \]
		For any $x_1>0$, there must exist $x_2<x_1$ such that $\hat{y}_1(x_2)>0$ and $\hat{y}'_1(x_2)>0$. If $\hat{y}_1$ is not increasing on $(x_2,\infty)$, then there would be a point $x_3>x_2$ such that $\hat{y}'_1(x_3)=0$ and $\hat{y}_1$ is increasing on $(x_2,x_3)$. This is impossible because \eqref{DCHE:ode} under the assumption that $a,d>0$, $b\ge 0$ implies that $\hat{y}''(x_3)>0$. Hence $\hat{y}_1(x)$ is positive and increasing on $(x_2,\infty)$ for some $x_2<x_1$, given any $x_1>0$. As $x_1$ is arbitrary, $\hat{y}_1(x)$ is increasing and positive for all $x\in (0,\infty)$.
		
		If $d<0$, we consider the formal solution of the form \eqref{DCHE:formal} with $\lambda=d$ and $\mu=c-2$. It follows immediately that $A_k$'s satisfy the recursion \eqref{DCHE:rec3}. Using the same arguments as in the previous case, we can show that $y_1$ determined by \eqref{DCHE:y1asy2} is positive and increasing on $(0,\infty)$ provided that $a>0, b\ge 0$.
	\end{proof}
	
	The formal adjoint of equation \eqref{DCHE:ode} is given by
	\begin{equation}\label{DCHE:adjoint}
	(x^2w)''-((cx+d)w)'-(ax+b) w =0.
	\end{equation}
	This equation is of the form as \eqref{DCHE:ode} with a new set of parameter values $\tilde a$, $\tilde b$, $\tilde c$, $\tilde d$, where
	\[ \tilde a=a,\quad \tilde b=b+c-2,\quad \tilde c=4-c,\quad \tilde d=-d .\]
	Equation \eqref{DCHE:adjoint} can be transformed to \eqref{DCHE:ode} by setting
	\[ y(x)=x^{2-c}\exp\left(\frac{d}{x}\right)\, w(x).\]
	In particular, the case $d<0$ of Theorem \ref{DCHE:y1} can be transformed to the case $d>0$.
	
	To our best knowledge, there is no mathematical software package programmed to evaluate the solutions $y_1$, $y_2$ introduced in Theorems \ref{DCHE:y1}, \ref{DCHE:y2}, respectively.
	We compute $y_1$ for $d>0$ as follows. We choose a positive integer $m$ and, according to \eqref{DCHE:y1asy}, approximate
	\[ \hat{y}_1(x_0) \approx \sum_{k=0}^{m-1} A_k x^k .\]
	This approximation will be good only for sufficiently large $x_0$.
	By differentiating term by term we also obtain an approximation for $y_1'(x_0)$.
	Using these approximations for $\hat{y}_1(x_0)$, $\hat{y}_1'(x_0)$ we solve the initial value problem for equation \eqref{DCHE:ode} using a suitable numerical procedure.
	If we need $\hat{y}_1(x)$ to high accuracy it is preferable to use a procedure like {\tt gear} implemented in Maple.
	
	The computation of $y_1$ for $d<0$ and $y_2$ is similar.
	
	\section{Joint distributions}

	Consider the stochastic differential equation (SDE) on an open interval $(0, \infty)$ 
	\begin{equation}\label{J:sde}
	dX_t=\beta(X_t)\df t + \gamma(X_t)\df B_t,
	\end{equation}
	with initial condition $X_0=x_0>0$ and its pathwise integral process $Y=\{Y_t, t\ge 0\}$ where
	\[ Y_t=\int_0^t X_s \df s .\]
	Then the pair $(X, Y)$ satisfies the SDE system
	\be \label{bivar1}
	\df X_t &=& \beta(X_t ) \df t+ \gamma(X_t) \df B_t,\\
	\df Y_t &=& X_t \df t, \label{bivar2}
	\ee under the probability measure $\Pro^{x_0}$ we have $\Pro^{x_0}(X_0=x_0, Y_0=0)=1$. 
	
	The work of H\"{o}rmander \cite{Hor} provides a condition for the smoothness of fundamental solutions to hypoelliptic second order differential equations, known as H\"{o}rmander's ``brackets condition" in the literature. Hairer \cite{Hai} provides an extension to parabolic equations in the context of transition probabilities for diffusions. Consider a second order differential operator
	\[\mathcal{P}=\sum^r_{j=1} \mathcal{X}^2_j+\mathcal{X}_0,\]
	where $ \mathcal{X}_0,  \mathcal{X}_1, \cdots,  \mathcal{X}_r$ denote first order homogeneous differential operators in an open set $\Omega \subset \R^n$ with $C^\infty$ coefficients and $c\in C^\infty(\Omega)$. Define a collection of vector fields $\mathcal{V}_k$ by
	\[\mathcal{V}_0=\{\mathcal{X}_i: i>0\}, \qquad \mathcal{V}_{k+1}=\mathcal{V}_k \cup \{[\mathcal{U}, \mathcal{X}_j]: \mathcal{U} \in \mathcal{V}_k \mbox{ and } j\ge 0\}.\]
	Also define the vector space $\mathcal{V}_k(x)=\mathrm{span}\{V(x): V\in \mathcal{V}_k\}.$ The parabolic H\''{o}rmander's condition holds if $\cup_{k\ge 1} \mathcal{V}_k(x)=\R$ for every $x\in \R$. The Lie bracket $[\cdot, \cdot]$ is given by $[\mathcal{X}_{j_i}, \mathcal{X}_{j_k}]=\mathcal{X}_{j_i}\mathcal{X}_{j_k}-\mathcal{X}_{j_k}\mathcal{X}_{j_i}.$ The infinitesimal generator satisfies the H\"ormander condition, then the Markov process admits a smooth density with respect to Lebesgue measure and the corresponding Markov semigroup maps bounded functions into smooth functions.
	
	Note that the infinitesimal generator of the two-dimensional process $(X_t, Y_t)$ in \eqref{bivar1} and \eqref{bivar2} is given by
	\[\mathcal{L}= \mathcal{X}^2_1+\mathcal{X}_0, \]
	where the two differential operators are defined by
	\by \mathcal{X}_0&=&\beta_\ast(x) \frac{\partial}{\partial x}+x \frac{\partial}{\partial y},\qquad \beta_\ast(x)=\beta(x)-\frac12 \gamma'(x) \gamma(x) \\
	\mathcal{X}_1&=&\gamma_\ast(x)\frac{\partial }{\partial x},\qquad \gamma_\ast(x)=\frac1{\sqrt{2}}\gamma(x).
	\ey It is easy to show that
	\begin{align*}
	[\mathcal{X}_0, \mathcal{X}_1]&=W(\beta_\ast, \gamma_\ast)\frac{\partial }{\partial x}-\gamma_\ast(x) \frac{\partial }{\partial y};\\
	[  [\mathcal{X}_0, \mathcal{X}_{1}], \mathcal{X}_0] \,&=W(W(\beta_\ast, \gamma_\ast),\beta_\ast)\frac{\partial }{\partial x}+[2\beta_\ast(x)\gamma'_\ast(x)-\gamma_\ast(x) \beta'_\ast(x)]\frac{\partial }{\partial y}
	\end{align*} where $W$ is the Wronskian. For example,  $[\mathcal{X}_0, \mathcal{X}_1]$ and $[  [\mathcal{X}_0, \mathcal{X}_{1}], \mathcal{X}_0]$ are linear independent at any point, unless there are points at which
	\be \label{check}
 [2\beta_\ast(x)\gamma'_\ast(x)-\gamma_\ast(x) \beta'_\ast(x)] W(\beta_\ast, \gamma_\ast) +\gamma_\ast(x) W(W(\beta_\ast, \gamma_\ast),\beta_\ast)=0.\ee
	

	Suppose that \eqref{check} is not true for any point $x>0$. Then H\"ormander condition implies that there exists a smooth joint density function of $(X_t, Y_t)$, denoted by $p(t,x,y)=p(t,x,y; x_0)$ , i.e.
	\[\Pro^{x_0}(X_t \in \df x , Y_t \in \df y)=p(t,x,y) \df x \df y.\] The corresponding Kolmogorov forward PDE for $p(t,x,y)$ is given by
	\[\frac{\partial p}{\partial t}=-\frac{\partial}{\partial x} (\beta(x) p)-\frac{\partial}{\partial y} (xp)+\frac12 \frac{\partial^2}{\partial x^2}(\gamma(x)^2 p),\] where the assumption $P^{x_0}(X_0=x_0, Y_0=0)=1$ implies that
	\[p(0,x,y)=\delta(x-x_0)\delta(y),\] where $\delta$ is the Dirac delta function for which
	\[\int^\infty_0 f(x) \delta(x) \df x=f(0),\qquad \int^\infty_{-\infty} \delta(x)\df x =1.\]
	We can solve this equation by using the Laplace transform
	\begin{equation}\label{J:laplace}
	q(x):=q(s,x,w;x_0)=\int_0^\infty \int_0^\infty \exp(-st)\exp(-w y)p(t,x,y)\,dt\,dy.
	\end{equation}
	For fixed $s,w>0$ we obtain the following ODE for $q(x)$
	\begin{equation}\label{J:ode}
	-\frac12 (\gamma(x)^2 q)''+(\beta(x) q)'+(wx+s) q=\delta(x-x_0) .
	\end{equation}
	This equation is the same as appears in the definition of a Green's function at $x=x_0$.
	We choose a fundamental system $q_1(x), q_2(x)$ of solutions of \eqref{J:ode} with right-hand side $0$.
	Then we write the solution $q(x)$ of \eqref{J:ode} in the form
	\begin{equation}\label{J:q}
	q(x)=\begin{cases} C_1 q_1(x) & \text{if $0<x\le x_0$,} \\
	C_2 q_2(x) & \text{if $x_0<x<\infty$.}
	\end{cases}
	\end{equation}
	Definition \eqref{J:laplace} shows that $q(x)\ge 0$ and $\int^{\infty}_{0} q(x) \df x \le 1/s<\infty$, so $q_1(x)$, $q_2(x)$ have to
	be integrable on $(0,x_0)$ and $(x_0,\infty)$, respectively.
	The constants $C_1$ and $C_2$ have to be chosen such that $q(x)$ is continuous at $x=x_0$, i.e.
	\be C_1 q_1(x_0)=C_2 q_2(x_0),\label{cond1}\ee and such that
	\be \frac12 \gamma(x_0)^2 (C_1 q'_1(x_0)-C_2 q'_2(x_0))=1.\label{cond2}\ee The latter equation follows from \eqref{J:ode} by integrating both sides $x_0-\epsilon$ to $x_0+\epsilon$, then letting $0<\epsilon\to 0$ and simplifying with \eqref{cond1}. Therefore, we obtain
	\begin{equation}\label{J:q2}
	q(x)=C\begin{cases} q_2(x_0)q_1(x) & \text{if $0<x\le x_0$,} \\
	q_1(x_0)q_2(x) & \text{if $x_0<x<\infty$,}
	\end{cases}
	\end{equation}
	where
	\begin{equation}\label{J:C}
	C=\frac2{\gamma(x_0)^2} \frac{1}{q_2(x_0)q_1'(x_0)-q_2'(x_0)q_1(x_0)}.
	\end{equation}
	
	We will consider two examples in this work. In this section
	we take $X$ to be geometric Brownian motion and reproduce the density $p(t,x,y)$, which is known from the work of Yor. 
	In the next section, we take $X$ to be geometric Brownian motion with affine drift where no previous results on the joint density $p(t,x,y)$ is known in the current literature.
	
	The solution of the SDE \eqref{J:sde} with $\beta(x)=(2\nu+2)x$, $\nu\in\R$, $\gamma(x)=2x$, $x_0=1$ is the geometric Brownian motion
	\[ X_t=\exp(2\nu t+2B_t)\]
	with integral $Y_t=\int_0^t X_s\,ds$.
	Then the homogeneous form of the ODE \eqref{J:ode} is
	\begin{equation}\label{J:ode2}
	-2(x^2 q)''+2(\nu+1)(xq)'+(wx+s) q =0.
	\end{equation}
	This is equation \eqref{DCHE:ode} with parameter values
	\[ a=\frac{w}{2},\quad b=\frac{s}{2}+\nu-1,\quad c=3-\nu,\quad d=0 .\]
	According to \eqref{DCHE:d0}, \eqref{J:ode2} has the fundamental system
	\[ q_1(x)=x^{\nu/2-1} I_\lambda(\sqrt{2wx}),\quad q_2(x)=x^{\nu/2-1} K_\lambda(\sqrt{2wx}) ,\]
	where
	\[ \lambda:=\sqrt{2s+\nu^2}. \]
	Note that, apart from constant multiples, $q_1(x)$ and $q_2(x)$ are the only solutions of \eqref{J:ode2} that are integrable on $(0,1)$ and $(1,\infty)$, respectively.
	Using the Wronskian Olver et al. \cite[10.28.2]{NIST}, we obtain that
	\[ q_1(x)q_2'(x)-q_1'(x)q_2(x)=-\frac12x^{\nu-3}. \]
	Setting $x=1,$ we obtain $C=1$ in \eqref{J:C}.
	Therefore, \eqref{J:q2} gives
	\begin{equation}\label{2:result}
	q(s,x,w)= x^{\nu/2-1} \begin{cases} K_\lambda(\sqrt{2w})I_\lambda(\sqrt{2wx}) &\text{if $0<x \le1$}, \\
	I_\lambda(\sqrt{2w})K_\lambda(\sqrt{2wx}) &\text{if $1<x<\infty$}.
	\end{cases}
	\end{equation}
	Essentially we have solved our problem. It remains to invert the Laplace transforms.
	We may use the formula Gradshteyn and Ryzhik~\cite[6.653]{GraRyz}
	\begin{equation}\label{2:grad}
	\int_0^\infty \exp\left(-\frac{z}2 -\frac{a^2+b^2}{2z  }\right)I_\lambda\left(\frac{ab}{z}\right)\frac{dz}{z} =2\begin{cases} I_\lambda(a)K_\lambda(b) & \text{if $0<a<b $}\\
	K_\lambda(a)I_\lambda(b) & \text{if $0<b<a$}
	\end{cases}
	\end{equation}
	which holds for $\lambda>-1$.
	If we set $a=\sqrt{2w x}$, $b=\sqrt{2w}$, $z=2w y$ we obtain
	\begin{equation}\label{2:result2}
	\int_0^\infty \exp(-st)p(t,x,y)\,dt =\frac1{2y} x^{\nu/2-1} \exp\left(-\frac{1+x}{2y}\right) I_\lambda\left(\frac{\sqrt{x}}{y}\right)
	\end{equation}
	which is a known result.
	We can also invert the Laplace transform with respect to $t$ employing the Hartmann-Watson density. Note that Yor obtained the same result using much more complex probabilistic arguments based on Lamperti's transformation and Girsanov change of measure in Yor \cite{Yor92}. A detailed account of Hartman-Watson density function can be found in Barrieu, Rounault and Yor~\cite{BRY}.

	\section{Joint distribution of Geometric Brownian motion with affine drift and its integral}\label{sec:asym}
	
	Let $\nu\in\R$, $x_0>0$, $\beta(x)=(2\nu+2)x+1$, $\gamma(x)=2x$. It is easy to verify that \eqref{check} does not hold for $x>0$. Then the solution of SDE \eqref{J:sde} is
	the geometric Brownian motion with affine drift
	\be \label{4.1}
	X_t=\exp(2\nu t+2 B_t)\left(x_0 +\int_0^t \exp(-(2\nu s+2B_s))\,ds\right) . \label{GBA:defX}\ee
	Its integral is
	\be Y_t=\int_0^t X_s\,ds .\label{GBA:defY}\ee

	Our goal is to determine the joint density function $p(t,x,y)$ of $(X_t,Y_t)$ by the method proposed in Section 3.
	The corresponding ODE \eqref{J:ode} is
	\begin{equation}\label{GBA:ode}
	-2(x^2 q)''+\left\{(2(\nu+1)x+1) q\right\}'+(wx+s) q =\delta(x-x_0).
	\end{equation}
	If we replace the right-hand side of \eqref{GBA:ode} by $0$, we obtain the doubly-confluent Heun equation \eqref{DCHE:ode} with parameter values
	\[a=\frac{w}2>0,\quad b=\frac{s}2+\nu-1,\quad c=3-\nu,\quad d=-\frac12 <0.\]
	We choose $q_j(x)=y_j(x)$, $j=1,2$, with $y_1$, $y_2$ introduced in Theorems \ref{DCHE:y1}, \ref{DCHE:y2}.
	Note that, apart from constant multiples, $q_2(x)$ is the only solution of \eqref{GBA:ode} which is integrable on $(x_0,\infty)$.
	The choice of $q_1(x)$ can be justified as follows. We have
	$0\le q(s,x,w)\le q(s,x,0)$,
	and $q(s,x,0)$ is the Laplace transform of the transition density function of $X_t$. It is known that $q(s,x,0)\to0$ as $x\to 0^+$. Therefore,
	$q(s,x,w)\to 0$ as $x\to 0^+$. Now, apart from constant multiples, $q_1(x)$ is the only solution of \eqref{GBA:ode} that tends to $0$ as $x\to 0^+$.
	
	We can then write $q(x)$ defined by \eqref{J:laplace} in the form \eqref{J:q2},
	where
	\[
	C=\frac1{2x_0^2} \frac{1}{q_2(x_0)q_1'(x_0)-q_2'(x_0)q_1(x_0)}.
	\]
	We cannot expect a simple formula for the Wronskian of solutions $q_1$ and $q_2$
	but we can compute $C$ numerically.
	
	We use the following algorithm to determine the joint density $p(t,x,y)$. In the first step, we use \eqref{DCHE:y1asy2} to approximate initial values for $q_1$ and \eqref{DCHE:y2asy} for $q_2$.
	In order to increase the accuracy of the computation, we define
	\[u(x)=\exp\left(  -\frac{d}{x}\right) x^{c-2} q_1(x),\]
	which satisfies the ODE
	\be x^2 u''(x)+[(4-c)x-d]u'(x)-(ax+b+c-2)u(x)=0.\label{GBA:uode}\ee
	Thus the asymptotics of $u(x)$ is determined by the power series in \eqref{DCHE:y1asy2}. We specify the number $m$ of terms in the partial sum approximation and use an algorithm to determine the small initial point $x_\ell>0$ such that the approximation error is less than $10^{-k}$, roughly $A_m x^m_l<10^{-k}$. Then we use numerical methods to find solutions to \eqref{GBA:uode} with initial conditions at $x_\ell$, which in turn produces the value of $q_1$ in the interval $(x_l,x_0)$.
	
	Similarly, we define
	\[w(z)=z^{c-1/2} \exp(2\sqrt{a}z)q_2(z^2).\]
	Then
	\[w''(z)+\left(-4\sqrt{a}+\frac{2d}{z^3}  \right) w'(z)+\left( \frac{-\frac34 -c^2+2c-4b}{z^2}-\frac{4d\sqrt{a}}{z^3}+\frac{d(1-2c)}{z^4} \right) w(z)=0.\]
	The asymptotics of $w(z)$ is given by the power series in \eqref{DCHE:formal}. For each specified number $m$ of terms in the partial sum approximation, we determine the large initial point $x_r$ such that the approximation error of $w(z)$ is less than $10^{-k}$, roughly $A_m x^{-m/2}_r<10^{-k}.$ We use numerical methods to solve the initial value problem for $w(z)$ with initial conditions determined by approximations, thereby leading to solution $q_2$ in $(x_0,x_r)$. Combining the computations of $q_1$ and $q_2$, we find the Laplace transform $q$ using \eqref{J:q2}.
	
	In the second step, we use two-dimensional Laplace inversion routines to find $p$ at various point of $(x,y)$. Details on various types of two-dimensional Laplace inversion can be seen in Abate and Whitt~\cite{AbaWhi06}. Here we obtain the results using the Talbot-Gaver-Stehfest algorithm
	\[p(t,x,y)=\frac{2\ln 2}{5t y} \sum^{M-1}_{k_1=0} \Re \left\{ \gamma_{k_1} \sum^{2M}_{k_2=1} \zeta_{k_2} q\left( \frac{\delta_{k_1}}{t},x, \frac{k_2 \ln 2}{y}\right) \right\},\]
	where
	\by \delta_0&=&\frac{2M}{5}, \delta_k=\frac{2k\pi}{5}\left(\cot\left(\frac{k\pi}{M}\right)+i\right), 0<k<M, \gamma_0=\frac{e^{\delta_0}}2,\\
	\gamma_k&=&\left[1+\frac{ik\pi}{M}\left(1+\cot^2\left(\frac{k\pi}{M}\right)\right)-i\cot\left(\frac{k \pi}{M} \right)\right] e^{\delta_k}, 0<k<M. \\
	\zeta_k&=&(-1)^{M+k} \sum^{k\wedge M}_{j=\lfloor (k+1)/2\rfloor} \frac{j^{M+1}}{M!} \dbinom{M}{j}\dbinom{2j}{j} \dbinom{j}{k-j}.\ey To test the accuracy of the results, we also use the Euler-Gaver-Stehfest algorithm
	\[p(t,x,y)=\frac{10^{M/3}\ln 2}{5 t_1 t_2} \sum^{2M}_{k_1=0} \eta_{k_1} \sum^{M-1}_{k_2=0} \Re \left\{  \gamma_{k_2} q\left( \frac{\beta_{k_1}}{t_1}, \frac{\delta_{k_2}}{t_2}\right)+\overline{\gamma}_{k_2}q\left( \frac{\beta_{k_1}}{t_1},\frac{\ol{\delta}_{k_2}}{t_2}   \right)\right\},\]
	where $\ol{\gamma}_k$ is the complex conjugate of $\gamma_k$ and
	\by
	\eta_0&=&\frac12, \eta_k=(-1)^k, 1\le k\le M, \eta_k=(-1)^k \sum^M_{i=k-M} \left(\begin{array}{c} M\\i\end{array}\right)2^{-M}, M<j\le 2M.
	\ey In this numerical example, we choose $\nu=1.2, M=7$. According to Abate and Whitt~\cite{AbaWhi06}, both algorithms are expected to be accurate up to $0.6 M \approx 4$ significant digits for ``good transforms". We show the values of $p(1,6,4)$ rounded up to $7$ digits in Table \ref{tbl:dblinv} using the two inversion algorithms with various choices of precision. They all agree up to four decimal places.
	
	\begin{table}[htb] \centering \begin{tabular}{|c|c||c|c|}
			\hline
			$m$ & $k$ &   Talbot-Gaver-Stehfest & Euler-Gaver-Stehfest  \\ \hline \hline
			30 & 15 &   $0.0047812$ & $0.0047684$ \\ \hline
			30 & 18&    $0.0047812$ & $0.0047684$\\ \hline
		\end{tabular}
		\caption{Joint density $p(1,6,4)$.} \label{tbl:dblinv}
	\end{table}
	
	The method can be easily extended to the geometric Brownian motion with affine drift
	\[X_t=\exp(2\nu t+2 B_t)\left(x_0 -\int_0^t \exp(-(2\nu s+2B_s))\,ds\right), \qquad Y_t=\int^t_0 X_s \df s,\] with $X_0=x_0>0$ and $Y_0=0$, in which case the Laplace transform $q(x)$ satisfies the ODE
	\[-2(x^2 q)''+\left\{(2(\nu+1)x-1) q\right\}'+(wx+s) q =0,
	\] which is also a special case of \eqref{DCHE:ode} with
	\[a=\frac{w}2>0, \quad b=\frac{s}2+\nu-1, \quad c=3-\nu, \quad d=\frac12 >0.\]
	We provide the results on $p(1,2,2)$ with $\nu=1.2$ and $M=7$ in Table \ref{tbl:dblinv2}.

	\begin{table}[htb] \centering \begin{tabular}{|c|c||c|c|}
			\hline
			$m$ & $k$ &   Talbot-Gaver-Stehfest & Euler-Gaver-Stehfest  \\ \hline \hline
			30 & 15 &   $0.016420$ & $0.016413$ \\ \hline
			30 & 18&    $0.016420$ & $0.016413$\\ \hline
		\end{tabular}
		\caption{Joint density $p(1,2,2)$.} \label{tbl:dblinv2}
	\end{table}
	
	\section{A diffusion process from Lamperti's transformation} 
	
	We recall from Yor~\cite{Yor92} that the geometric Brownian motion and its time-integral are connected through Lamperti's transformation, i.e.
	\[ \exp\{B^{(\nu)}_t\}=\rho^{(\nu)}_{A^{(\nu)}_t},\qquad t \ge 0,\]
	where $B^{(\nu)}_t$, $A_t^{(\nu)}$ are defined by \eqref{J:BA}, and $\rho$ is a Bessel process with index $\nu$ starting from $1$. Using the same idea, we can connect the geometric Brownian motion with affine drift and its time-integral through another diffusion process.
	
	\bt
	Let $X_t$ be the solution of the SDE \eqref{J:sde} with $\beta(x)=\mu x+1$, $\gamma(x)=\sigma x$, where $\mu\in\R$, $\sigma>0$ and $X_0=x_0>0$.
	Then
	\be X_t =  Z_{\int^t_0 X_s \df s},\qquad t \ge 0,\label{idt}\ee
	where $Z$ is a diffusion process determined by the following SDE
	\be \df Z_t =\left( \mu +\frac1{Z_t} \right)\df t +\sigma \sqrt{Z_t} \df B_t. \label{sde}\ee
	\et
	\begin{proof}
		Let $\tau(t)=\inf\{u: \int^u_0 X_s \df s >t\}$. Note that $\int^t_0 X_s \df s$ is a process with continuous sample path.
		Then $\int_0^{\tau(t)} X_s\,ds=t$. Thus,
		\[X_{\tau(t)}=x_0+\int^{\tau(t)}_0 (\mu X_s +1) \df s+\sigma \int^{\tau(t)}_0 X_s \df B_s=x_0+\mu t +\int^{\tau(t)}_0 \df s+\sigma \int^{\tau(t)}_0 X_s \df B_s.\]
		Using the time change formula for Ito integrals (Oksendal \cite[Theorem 8.5.7, p156]{Oks}), we obtain
		\[\int^{\tau(t)}_0 X_s \df B_s=\int^t_0 X_{\tau(r)} \sqrt{\tau'(r)} \df W_r,\]
		where $W_t=\int^{\tau(t)}_0 \sqrt{X_s} \df B_s$ is also a Brownian motion. Note that $\tau'(r)=1/X_{\tau(r)}$ using the derivative of inverse function. Thus,
		\[\int^{\tau(t)}_0 X_s \df B_s=\int^t_0  \sqrt{X_{\tau(r)}} \df W_r.\]
		We also have
		\[ \int^{\tau(t)}_0 \df s = \int^t_0 \tau'(r) \df r=\int^t_0 \frac1{X_{\tau(r)}} \df r.\]
		Let $Z_t:=X_{\tau(t)}$ for $t\ge 0$. Then,
		\[ Z_t=Z_0+ \mu t+ \int^t_0 \frac1{Z_r} \df r+ \sigma \int^t_0  \sqrt{Z_r} \df W_r.\]
		Thus the claim follows immediately after reversing the time change.
	\end{proof}
	
	To the authors' best knowledge, this diffusion process $\{Z_t, t \ge 0\}$ was not previously studied in the probability literature. When starting from a positive initial value and $\mu<0$, this process is always positive and possesses a mean-reverting property, which is a desirable property for modeling many physical phenomenons. If we remove the Brownian perturbation by setting $\sigma=0$, then
	\[Z_t=-\frac1{\mu} \left( 1+W(e^{-C\mu^2-\mu^2 t-1})\right),\] where $W$ is the Lambert W-function.

	The results from Section 2 make it possible to compute the transition density of the diffusion process $Z_t$ without having to use the often expensive Monte Carlo simulations. An account of the theory of time-homogeneous diffusion processes can be found in Borodin and Salminen~\cite[Chapter 2]{BS}.
	In this case, it is easier to compute the transition density function $\mathfrak{p}$ with respect to speed measure, i.e.
	\be \Pro^{z_0}(Z_t \in \df z)=p(t,z_0,z) \df z=\mathfrak{p}(t,z_0,z)\mathfrak{m}(z) \df z,\label{Zden}\ee
	where $Z_0=z_0>0$. The speed density and scale density of the diffusion process \eqref{sde} are given by
	\[\mathfrak{m}(z)=\frac{2}{\sigma^2} z^{\frac{2\mu}{\sigma^2}-1} \exp\left\{ \frac{2}{\sigma^2}\left(1-\frac1z\right) \right\},\quad \mathfrak{s}(z)=z^{-\frac{2\mu}{\sigma^2}}\exp\left\{  \frac2{\sigma^2}\left( \frac1y-1 \right) \right\}.\] 
	Then $p$ satisfies the forward Kolmogorov equation
	\[ \frac{\partial p}{\partial t} =-\frac{\partial}{\partial z}(\beta(z) p)+\frac12\frac{\partial^2}{\partial z^2} (\gamma(z)^2 p),\]
	while $\mathfrak{p}$ solves the backward Kolmogorov equation.
	\[ \frac{\partial\mathfrak{p}}{\partial t} =\beta(z) \frac{\partial \mathfrak p}{\partial z}+\frac12\gamma(z)^2 \frac{\partial^2 \mathfrak p}{\partial z^2},
	\]
	where $\beta(z)=\mu +1/z$, $\gamma(z)=\sigma\sqrt{z}$.
	Consider the Laplace transform
	\[\mathfrak{q}(s,z_0,z)=\int^\infty_0 e^{-st} \mathfrak{p}(t,z_0,z) \df t.\]
	Then $\mathfrak{q}(z):=\mathfrak{q}(s,z_0,z)$ satisfies the ODE
	\be \frac{\sigma^2}{2} z \mathfrak{q}''(z)+\left( \mu+\frac1{z}  \right) \mathfrak{q}'(z)-s\mathfrak{q}(z)=0,\qquad z>0.\label{green}\ee
	This ODE is a special case of \eqref{DCHE:ode} with parameter values
	\[a=\frac{2s}{\sigma^2}>0,\quad b=0,\quad c=\frac{2\mu}{\sigma^2},\quad d=\frac2{\sigma^2}>0.\]
	
	Let $\mathfrak{q}_1$, $\mathfrak{q}_2$ be the positive monotone solutions $y_1$, $y_2$ introduced in Theorems \ref{DCHE:y1}, \ref{DCHE:y2}, respectively.
	Then
	\be
	\mathfrak{q}(s,z_0,z)=\frac1{w(s,z)}\begin{cases} \mathfrak{q}_1(z) \mathfrak{q}_2(z_0) & \text{if $0<z<z_0$},\\
		\mathfrak{q}_1(z_0) \mathfrak{q}_2(z) & \text{if $z_0< z$,}
		\label{greensoln}
	\end{cases}
	\ee
	where $w(z)$ denotes the Wronskian
	\[w(z)=\frac1{\mathfrak{s}(z)} [\mathfrak{q}'_1(z) \mathfrak{q}_2(z)-\mathfrak{q}_1(z) \mathfrak{q}'_2(z)].\]
	The following algorithm is used in this numerical example to calculate the transition
	density $p(t,z_0,z)$.
	
	\bn
	\item Approximate the values of $\mathfrak{q}_1$ and $\mathfrak{q}_2$ at some initial points by the asymptotics \eqref{DCHE:y1asy} and \eqref{DCHE:y2asy}. Take the approximations as initial conditions and use numerical methods for initial value problems to determine $\mathfrak{q}_1$ and $\mathfrak{q}_2$.
	\item Substitute the values of $\mathfrak{q}_1$ and $\mathfrak{q}_2$ in \eqref{greensoln} to determine values of the Laplace transform $\mathfrak{q}$. Use numerical methods for inverting the Laplace transform to evaluate $\mathfrak{p}.$
	\en
	\begin{figure}[ht!]
		\centering
		\includegraphics[width=0.6\textwidth]{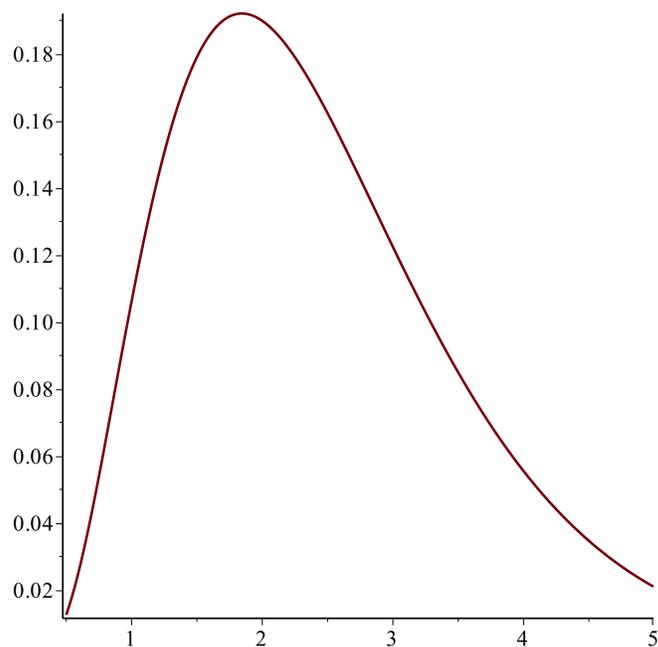}
		\caption{Probability density function $p(1,1,z)$ of the process $Z$}
		\label{fig:den}
	\end{figure}
	
	In our example, we use the parameters $\mu=0.8, \sigma=1$ and initial position $z_0=1$. In the first step, we find approximation of $\mathfrak{q}_1(1/25)$ and $\mathfrak{q}_2(25)$ using the first $40$ terms in the asymptotics \eqref{DCHE:y1asy} and \eqref{DCHE:y2asy}. Then we use Maple's numerical method {\tt gear} to determine solutions to the initial value problems, which yield solutions of $\mathfrak{q}_1$ and $\mathfrak{q}_2$ in the interval $(1/25, 25)$. To test the accuracy of the results, we use two numerical methods for inverting Laplace transforms, namely, the Gaver-Stehfest algorithm
	\[\mathfrak{p}(t,z_0,z)\approx \sum^n_{k=1} (-1)^{n-k}\frac{k^n}{k!(n-k)!} \td{f}_k(t,z_0,z),\]
	where
	\[\td{f}_k(t,z_0,z)=\frac{\ln 2}{t} \frac{(2k)!}{k!(k-1)!} \sum^{k}_{h=1} (-1)^h \left( \begin{array}{c} k\\h \end{array} \right) \mathfrak{q}\left((h+k)\frac{\ln 2}{t},z_0,z\right),\]
	and the Euler algorithm (cf. Abate and Whitt \cite{AbaWhi95})
	\[\mathfrak{p}(t,z_0,z)\approx \sum^m_{k=0} \left( \begin{array}{c} m\\ k \end{array} \right) 2^{-m} s_{n+k}(t),\]
	where
	\[s_n(t)=\frac{e^{A/2}}{2t} \Re (\mathfrak{q})\left(\frac{A}{2t},z_0,z\right)+\frac{e^{A/2}}{t} \sum^n_{k=1} (-1)^k \Re(\mathfrak{q})\left(\frac{A+2k \pi i}{2t},z_0,z\right).\]
	Because the two algorithms are sufficiently different, we can verify the accuracy of our numerical method by observing the results on probability density function $\mathfrak{p}$ from both algorithms which agree up to at least four decimal places in all calculations. As recommended by Abate and Whitt~\cite[page 38]{AbaWhi95}, we choose $A=18.4$. Figure \ref{fig:den} displays the graph of the probability density function $p(1,1,z)$ of the diffusion process $Z$.

	\section{Joint Laplace transform of the geometric Brownian motion with affine drift and its time-integral }
	{ 
	An alternative form of the GBM with affine drift to \eqref{4.1} is given by 
	\be \df X_t=(rX_t -\delta) \df t+\sigma X_t \df B_t,\qquad X_0=x_0>0.\label{AsianSDE}\ee Since such a model typically arises in the context of option pricing in finance literature, we use financial interpretation of model parameters. Assume that $X$ represents the dynamics of stock prices. Dividends are paid out at a constant rate $\delta>0$ per time unit and the risk-free continuously compounding interest rate is given by $r$ per time unit. 	It should be pointed out that in this model stock prices may hit zero at one point, which is considered the time of ruin, prior to a fixed option maturity date $T$. In this case, the process is assumed to be absorbed at zero once it hits zero. 
One can also find that  the SDE \eqref{AsianSDE} is a special case of \eqref{J:sde} with $\gamma(x)=\sigma x$ and $\beta(x)=rx-\delta$.  

	We focus on the joint Laplace transform of the process and its time-integral which plays an important role in an application to Asian option.  The Laplace transform is defined by 
\be h(x_0,T; w):=\Exp_{x_0}\left[ e^{-w\int^T_0 X_t \df t-\lambda X_T}  \right],\label{hrep}\ee where $\lambda, w>0$. Then it follows immediately from the Feymann-Kac formula that $h$ satisfies the PDE for $0<t<T$ and $x>0$
\be \frac{\partial h}{\partial t}+w x h=( rx-\delta) \frac{\partial h}{\partial x} +\frac{\sigma^2}{2} x^2 \frac{\partial^2 h}{\partial x^2}, \label{pdeh}\ee subject to 
\[ \begin{cases} h(x_0, 0)=e^{-\lambda x_0}, & x_0>0; \\ h(0, t)=1, & 0<t<T.\end{cases} \]
Consider the Laplace transform $\tilde{h}(x, s):=\int^\infty_0 e^{-st} h(x,t) \df t$ for some $s>0$. Taking Laplace transforms with respect to $t$ on both sides of \eqref{pdeh} we obtain the following ODE of $\tilde{h}(x)=\tilde{h}(x,s)$ for $x>0$
\begin{align}\label{odeh}
\frac{\sigma^2}{2} x^2  \tilde{h}''(x) +(rx-\delta) \tilde{h}'( x) -(s  +wx) \tilde{h}(x)=-e^{-\lambda x_0}, 
\end{align} 
subject to boundary conditions
\be \label{bd1} \tilde{h}(0)=\frac{e^{-\lambda x_0}}{s},\\
\label{bd2} \lim_{x\to \infty}\tilde{h}(x)=0,
\ee where the second condition comes from the interpretation of probabilistic representation \eqref{hrep} as $x_0\rightarrow \infty$.

We first prove the existence of a unique solution to the boundary value problem and then present the asymptotics of its solution. Let us consider the linear differential equation
	\begin{equation}\label{ode1}
		x^2y''+(cx+d)y'-(ax+b)y =0 ,\quad x>0,
	\end{equation}
	where $a>0$, $b>0$, $d<0$ and $c\in\R$. 
	Based on the discussion of section 2, we know that  there is a positive number $v$ such that the Wronskian $W$ of $y_1, y_2$ is 
	\begin{equation}\label{W}
	W(x)=-v x^{-c} \exp\tfrac{d}{x}<0 .
	\end{equation}
	Then we will further
consider the inhomogeneous linear differential equation 
	\begin{equation}\label{ode2}
		x^2y''+(cx+d)y'-(ax+b)y =f(x) ,\quad x>0.
	\end{equation}
	We have the following result.
	\begin{theorem}\label{t1}
		Let $f$ be a bounded real-valued continuous function on $(0,\infty)$, and let $\theta\in\R$.
		Then there exists a unique solution $y$ of \eqref{ode2} with the properties
		\begin{equation}\label{prop}
			\lim_{x\to0^+} y(x)=\theta,\quad \lim_{x\to\infty} y(x)=0. 
		\end{equation} 
	\end{theorem}
	\begin{proof}
		We first prove uniqueness. Let $y, \tilde y$ be two solutions of \eqref{ode2} satisfying \eqref{prop}. Then 
		$y-\tilde y$ is a solution of \eqref{ode1} with limit $0$ as $x\to0^+$ and $x\to\infty$. Therefore, $y-\tilde y$
		must be a multiple of $y_1$ and $y_2$. But $y_1, y_2$ are linearly independent by \eqref{W}.
		Therefore, $y=\tilde y$.  
		
		We now prove existence. It is enough to consider $\theta=0$ because if $y$ is a solution of \eqref{ode2} satisfying \eqref{prop} with $\theta=0$ then 
		$y(x)+\theta \frac{y_2(x)}{y_2(0)}$ is the desired solution of \eqref{ode2} satisfying \eqref{prop}.
	It is claimed that the solution (for $\theta=0$) is given by 
		\begin{equation}\label{var}
			y(x)=\left(\int_x^\infty \frac{y_2(t) f(t)}{t^2W(t)}\,dt\right) y_1(x)+\left(\int_0^x \frac{y_1(t)f(t)}{t^2 W(t)}\,dt\right)y_2(x) .
		\end{equation}
	
	The first integral in \eqref{var} exists for $x>0$.
		We use that $f$ is bounded, \eqref{DCHE:y2asy} and \eqref{W} to estimate  
		\begin{equation}\label{eq1}
			\left|\frac{y_2(t) f(t)}{t^2W(t)}\right| \le C t^{-\frac74+\frac{c}{2}}\exp(-2\sqrt{at}),\quad t\ge t_0>0.
		\end{equation}
		In this proof $C$ denotes a constant which may have different values in different inequalities.
		 The second integral in \eqref{var} exists for $x>0$.
		This follows from
		\begin{equation}\label{eq2}
			\left|\frac{y_1(t)f(t)}{t^2 W(t)}\right| \le C ,\quad 0<t<t_0,
		\end{equation} 
		where we used \eqref{DCHE:y1asy2} and \eqref{W}.
		It now follows that $y$ defined by \eqref{var} is a solution of \eqref{ode2}. 
		We show that the first term on the right-hand side of \eqref{var} converges to $0$ as $x\to\infty$.
		It is easy to find that, for every fixed $\alpha\in\R$,  
		\begin{equation}\label{eq3}
			\int_u^\infty s^\alpha e^{-s}\,ds\le C u^\alpha e^{-u} ,\quad u\ge u_0>0 .
		\end{equation}
		If we substitute $s=2\sqrt{at}$ we obtain
		\[
		\int_x^\infty t^{\frac12(\alpha-1)}\exp(-2\sqrt{at}) \,dt \le C x^{\frac{\alpha}{2}} \exp(-2\sqrt{ax}),\quad x\ge x_0>0 .
		\]
		Therefore, using this inequality with $\alpha=-\frac52+c$ and combining with \eqref{eq1}, 
		\[ \int_x^\infty \left|\frac{y_2(t) f(t)}{t^2W(t)}\right|\,dt\le Cx^{-\frac54+\frac{c}2}\exp(-2\sqrt{ax}),\quad x\ge x_0.\]
		
		Following similar calculation to that in the proof of Theorem \ref{DCHE:y2}, we obtain two independent solutions with known asymptotics near $x=\infty$, one of which is $y_2$ determined by the asymptotics in \eqref{DCHE:y2asy} and the other, denoted by $\hat{y}_2$, determined by the asymptotics
		\begin{equation}\label{y2hat}
		\hat{y}_2(x) \sim x^{1/4-c/2} e^{2\sqrt{ax}} \sum^\infty_{k=1} A_k x^{-k/2}, \qquad \mbox{ as } x \rightarrow \infty,
		\end{equation}
		 with $A_0=1$ and $A_k$'s determined by the recursive relation
		\[A_k=\frac{(k+1/2-c)(k-3/2+c)-4b}{4\sqrt{a}k}A_{k-1}+\frac{d}{k}A_{k-2},\] with the understanding that $A_{-1}=0$. 
		It follows from \eqref{DCHE:y2asy}, \eqref{y2hat} that  
		\begin{equation}\label{est1}
			0<y_1(x)\le C x^{\frac14-\frac{c}{2}}\exp(2\sqrt{ax}),\quad x\ge x_0 .
		\end{equation} 
		Therefore, 
		\[ \int_x^\infty \left|\frac{y_2(t) f(t)}{t^2W(t)}\right|\,dt\,\, y_1(x)\le C x^{-1},\quad x\ge x_0 .
		\]

		We show that the second term in \eqref{var} converges to $0$ as $x\to\infty$.
		It is easy to see that, for fixed $\alpha\in\R$, 
		\begin{equation}\label{eq4}
			\int_1^u s^\alpha e^s\,ds\le C u^\alpha e^u ,\quad u\ge u_0>0 .
		\end{equation}
		If we substitute $s=2\sqrt{at}$ we obtain
		\begin{equation}\label{eq5}
			\int_1^x t^{\frac12(\alpha-1)}\exp(2\sqrt{at}) \,dt \le C x^{\frac{\alpha}{2}} \exp(2\sqrt{ax}),\quad x\ge x_0 .
		\end{equation}
		From \eqref{W} and \eqref{est1}, we have,
		\[ \left|\frac{y_1(t) f(t)}{t^2W(t)}\right| \le C t^{-\frac74+\frac{c}{2}} \exp(2\sqrt{at}),\quad t\ge t_0 .\]
		Therefore, with $\alpha=-\frac52+c$ in \eqref{eq5},
		\[ \int_0^x \left|\frac{y_1(t)f(t)}{t^2W(t)}\right|\,dt \le Cx^{-\frac54+\frac{c}2}\exp(2\sqrt{ax}),\quad x\ge x_0.\]
		Using \eqref{DCHE:y2asy},
		\[ \int_0^x \left|\frac{y_1(t)f(t)}{t^2W(t)}\right|\,dt\,\, y_2(x)\le C x^{-1},\quad x\ge x_0 .\]

In what follows, we aim to prove that the first term on the right-hand side of \eqref{var} converges to $0$ as $x\to0^+$.
		It follows from \eqref{DCHE:y1asy}, \eqref{DCHE:y1asy2} that $y_2(t)$ is bounded as $t\to 0^+$. Therefore, we have
		\[ \left|\frac{y_2(t)f(t)}{t^2W(t)}\right|\le C t^{c-2}\exp(-\tfrac{d}{t}), \quad 0<t<t_0.\]
		When we substitute $s=-\frac{d}{t}$, $\alpha=-c$ in \eqref{eq4} we obtain
		\[ \int_x^\infty \left| \frac{y_2(t)f(t)}{t^2W(t)}\right|\,dt \le C x^{c}\exp(-\tfrac{d}{x}),\quad 0<x<x_0 .\]
		Therefore, by \eqref{DCHE:y1asy2},
		\[  \int_x^\infty \left| \frac{y_2(t)f(t)}{t^2W(t)}\right|\,dt\,\, y_1(x)\le Cx^2,\quad 0<x<x_0 .\]
		Similarly,  the second term on the right-hand side of \eqref{var} converges to $0$ as $x\to0^+$.
		since from \eqref{eq2}, we get
		\[ \int_0^x \left|\frac{y_1(t)f(t)}{t^2W(t)}\right|\,dt\,\,y_2(x)\le C x ,\quad 0<x<x_0 .\] 
		This completes the proof.
	\end{proof}

	We will shall the following lemma, whose proof is largely based on integration by parts and hence omitted from the paper.
	
	\begin{lemma}\label{l1}
		Let $\alpha\in\R$. We have the following asymptotics
		\begin{eqnarray}
			\int_u^\infty s^\alpha e^{-s}\,ds &\sim& u^\alpha e^{-u} \sum_{k=0}^\infty \alpha_k u^{-k}\quad\text{as $u\to\infty$},\label{asy1}\\
			\int_1^u s^\alpha e^s\,ds&\sim & u^\alpha e^u \sum_{k=0}^\infty (-1)^k \alpha_k u^{-k}\quad\text{as $u\to\infty$}\label{asy2},
		\end{eqnarray} where $\alpha_p:=\alpha(\alpha-1)\dots(\alpha-p+1).$
	\end{lemma} 
	
	\begin{theorem}\label{t2}
		The unique solution $y(x)$ of \eqref{ode2}, \eqref{prop} with $f(x)=-u$ satisfies 
		\begin{equation}\label{asy3}
			y(x)\sim \sum_{k=0}^\infty C_k x^k\quad \text{as $x\to 0^+$},
		\end{equation}
		where $C_0=\theta$, $dC_1=b\theta-u$ and, for $k\ge 2$, 
		\begin{equation}\label{rec1}
			kd C_k=(b-c(k-1)-(k-1)(k-2))C_{k-1}+ a C_{k-2} .
		\end{equation} 
		Moreover,
		\begin{equation}\label{asy4}
			y(x)\sim \sum_{k=0}^\infty D_kx^{-k-1}\quad \text{as $x\to\infty$},
		\end{equation}
		where $D_0=ua^{-1}$ and, for $k\ge 1$,
		\begin{equation}\label{rec2}
			aD_k=(k(k+1)-ck-b)D_{k-1}+d(k-1)D_{k-2},\quad D_{-1}:=0 .
		\end{equation} 
	\end{theorem}
	\begin{proof}
		We consider first the case $\theta=0$, $u=1$. Then the solution $y(x)$ is given by \eqref{var} with $f(x)=-1$.

		We show that the first term on the right-hand side of \eqref{var} with $f(x)=-1$ admits an asymptotic expansion of the form \eqref{asy3}.
		It follows from \eqref{DCHE:y1asy}, \eqref{DCHE:y1asy2} that
		\begin{equation}\label{asy5}
			y_2(t)=\sum_{k=0}^{K-1} E_k t^k+O(t^K)\quad\text{as $t\to0^+$} .
		\end{equation} 
		Therefore, 
		\[ \frac{y_2(t)(-1)}{t^2W(t)}=v^{-1}t^{c-2}\exp(-\tfrac{d}{t})\left(\sum_{k=0}^{K-1} E_k t^k+O(t^K)\right)\quad\text{as $t\to0^+$} .\]
		Now we use \eqref{asy2} with the substitution $s=-\frac{d}{t}$. Then we obtain
		\[ \int_x^\infty \frac{y_2(t)(-1)}{t^2W(t)}\,dt=x^c\exp(-\tfrac{d}{x})\left(\sum_{k=0}^{K-1} F_k x^k+O(x^K)\right)\quad\text{as $x\to0^+$} .\]
		Multiplying this integral by $y_1(x)$ and using \eqref{DCHE:y1asy2}, we obtain the desired asymptotic expansion.
		
		We show that the second term on the right-hand side of \eqref{var} with $f(x)=-1$ admits an asymptotic expansion of the form \eqref{asy3}.
		From \eqref{DCHE:y1asy2} and \eqref{W} we get
		\[ \frac{y_1(t)(-1)}{t^2 W(t)}=v^{-1}\left(\sum_{k=0}^{K-1} A_k t^k+O(t^K)\right) \quad\text{as $t\to0^+$} .\]
		Integrating this equation on both sides from $t=0$ to $t=x$, multiplying by $y_2(x)$ and using \eqref{asy5}
		we arrive at the desired asymptotic expansion.

		By differentiating \eqref{var}, we show that $y'(x)$ and $y''(x)$ have asymptotic expansions as  $x\to0^+$ which are obtained  by differentiating the asymptotic expansion of $y(x)$ term-by-term.
		If $y(x)$ is the solution for $\theta=0$ and $u=1$, then the solution $\tilde y(x)$ for real $\theta,u$ is
		\begin{equation}\label{general}
			\tilde y(x)=u y(x)+\theta\frac{y_2(x)}{y_2(0)} .
		\end{equation} 
		Therefore, $\tilde y(x)$ also admits an asymptotic expansions of the form \eqref{asy3}, and it may be differentiated.
		Substituting the  expansions for $\tilde y(x), \tilde y'(x), \tilde y''(x)$ in \eqref{ode2} with $f(x)=-u$ and comparing coefficients, we obtain the stated recursion for the coefficients $C_k$.

		Now we show that the first term $g(t)$ on the right-hand side of \eqref{var} with $f(x)=-1$ admits an asymptotic expansion of the form \begin{equation}\label{asy6}
			g(t)\sim \sum_{k=2}^\infty E_k x^{-k/2} \quad \text{as $x\to\infty$.}
		\end{equation}
		We know from \eqref{DCHE:y2asy} that 
		\[ y_2(t)=t^{\frac14-\frac{c}{2}}\exp(-2\sqrt{at})\left(\sum_{k=0}^{K-1} A_k t^{-k/2}+O(t^{-K/2})\right) \quad\text{as $t\to\infty$.}\]
		Therefore,
		\[ \frac{y_2(t)(-1)}{t^2W(t)}=t^{-\frac74+\frac{c}{2}}\exp(-2\sqrt{at})\left(\sum_{k=0}^{K-1} F_k t^{-k/2}+O(t^{-K/2})\right)
		\quad\text{as $t\to\infty$.}
		\]
		We now use \eqref{asy1} with the substitution $s=2\sqrt{at}$. Then we obtain
		\[ \int_x^\infty \frac{y_2(t)(-1)}{t^2W(t)}\,dt= x^{-\frac54+\frac{c}{2}}\exp(-2\sqrt{ax})\left(\sum_{k=0}^{K-1} G_k x^{-k/2}+O(x^{-K/2})\right) \quad\text{as $x\to\infty$} .\]
		If we multiply this integral by $y_1(x)$ and use \eqref{DCHE:y2asy}, \eqref{y2hat}, we obtain an asymptotic expansion of the 
		desired form \eqref{asy6}.
		
	In a similar way we show that the second term on the right-hand side of \eqref{var} with
		$f(x)=-1$ admits an asymptotic expansion of the form \eqref{asy6}.
		Using \eqref{general} we see that $\tilde y(x)$ also admits an asymptotic expansion of the form \eqref{asy6}, and we 
		obtain asymptotic expansions for $\tilde y'(x), \tilde y''(x)$ by  differentiating the asymptotic expansion for $\tilde y(x)$ term-by-term.
		By substituting these expansions in \eqref{ode2} with $f(x)=-u$ and compare the coefficients, we notice that the coefficients of the terms $x^{-k/2}$ with odd $k$ must vanish. Thus we obtain the recursion \eqref{rec2} and the initial value $D_0$. 
	\end{proof}

Assume that the dynamics of	stock prices with divided paying are driven by the geometric Brownian motion with affine drift defined by \eqref{AsianSDE} under the risk-neutral measure.
While the pricing of Asian options has been studied extensively, all existing literature are restricted to the case of options on non-dividend paying stocks. 
Applying the analysis of the doubly-confluent Heun equation, we can now take into account the more general case of options on dividend-paying stocks.
We are interested in the no-arbitrage price of a $T$-period Asian call option with the strike price of $K$, i.e.
\[\mathrm{AC}:=\Exp_{x_0}\left[e^{-rT}\left(\frac1T \int^T_0 X_t \df t -K\right)_+\right]=\frac1T e^{-rT}\Exp_{x_0}[(Y_T-K^\ast)_+],\]
where \[Y_T=\int^T_0 X_t  \df t,\qquad K^\ast=KT>0.\] For notational brevity, we assume that the expectation is taken under the risk-neutral probability measure for which
\[\Pro[X_0=x_0, Y_0=0]=1.\]
We can compute the following Laplace transform w.r.t. $K^\ast$. Since the integrands are nonnegative, we exchange the order of integration and obtain
\be
\int^\infty_0 e^{-wK^\ast} \Exp_{x_0}[(Y_T-K^\ast)_+] \df K^\ast&=&\Exp_{x_0}\left[ \int^{Y_T}_0 e^{-wK^\ast}(Y_T-K^\ast) \df K^\ast \right]\nonumber\\
&=& \frac1w \Exp_{x_0}(Y_T)-\frac1{w^2}+\frac1{w^2}\Exp_{x_0} [e^{-wY_T}]. \label{lap}
\ee
An analytic expression for the first term in \eqref{lap} is already obtained in Feng and Volkmer \cite[Proposition 3.4]{FV2}.\footnote{In Feng and Volkmer \cite{FV2}, this quantity was represented as the time-integral of the process $X_t$ up to the earlier of the first time it hits zero and a fixed time $T$, i.e. $\Exp\left[ \int^{\tau \wedge T}_0 X_t   \df t \right]$, where $\tau:=\inf\{t: X_t\le 0\}$.} 
Therefore, the only unknown quantity to be determined is the third term in \eqref{lap}.
 Once efficient algorithms for computing both terms are obtained, we would invert the Laplace transform with respect to $K^\ast$ to determine the price of Asian option. 
 One can easily find that it is a special case of $h$ defined by \eqref{hrep} by taking $\lambda=0$. 
Since the rest of computation for pricing Asian options is irrelevant to the discussion of asymptotics of Heun equation, we shall only focus on the computation of $\tilde{h}$, i.e. the Laplace transform of $h$, which can be derived by letting $\lambda=0$ in \eqref{odeh}.


Returning to the problem on the Laplace transform, we observe that \eqref{odeh} is a special case of the inhomogeneous Heun's equation \eqref{ode1} with $f(x)=-u$ and boundary conditions \eqref{prop}, 
	\[ a=\frac{2w}{\sigma^2},\quad b=\frac{2s}{\sigma^2},\quad c=\frac{2r}{\sigma^2},\quad d=-\frac{2\delta}{\sigma^2},\quad
	u=\frac{2}{\sigma^2},\quad\theta=\frac{1}{s}.\]

Although a numerical algorithm can be applied directly to solve the boundary value problem \eqref{odeh} with \eqref{bd1} and \eqref{bd2}, one would have to truncate the domain $(0, \infty)$ to some finite interval for practical reason. Then the asymptotics of the particular solution in \eqref{asy3} and \eqref{asy4} become useful to determine values at boundary points.

Here we provide a numerical example to show the solution to the boundary value problem using the asymptotics. The following set of parameters are used for computation.
\[s=2, w=1, \sigma=0.3, r=0.05, \delta=0.02.\]
To avoid the singularities $x=0$ and $x=\infty$, we consider the left-end-point to be $x_0=0.01$ and the right-end-point to be $x_1=100$. We use $30$ terms of the asymptotic formulas \eqref{asy3} and \eqref{asy4} to find $\tilde{h}(x_0)=0.499062426649333$ and $\tilde{h}(x_1)=0.00980754872574340.$ Finally, we use Maple's own default BVP solver to determine the numerical solution, which is shown in Figure \ref{fig:bvp}.	
	
	\begin{figure}[ht!]
		\centering
		\includegraphics[width=0.5\textwidth]{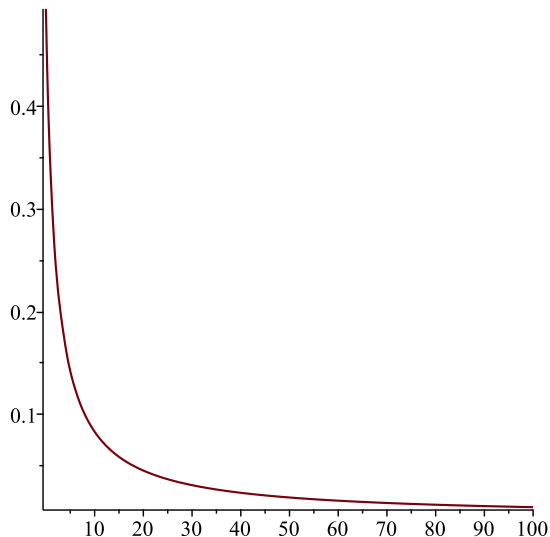}
		\caption{Solution to the boundary value problem \eqref{odeh}.}
		\label{fig:bvp}
	\end{figure}
}	
%

\end{document}